\documentclass[onecolumn,12pt,letterpaper,unpublished,nopdfoutputerror,accepted=2026-10-08]{quantumarticle}
\usepackage[utf8]{inputenc}
\usepackage{amsmath}
\usepackage{hyperref}
\usepackage{cleveref}
\usepackage{amssymb}
\usepackage{physics}
\usepackage{amsthm}
\usepackage{mathrsfs}

\DeclareMathOperator{\rect}{rect}
\DeclareMathOperator{\erfc}{erfc}
\DeclareMathOperator{\sinc}{sinc}

\usepackage[OT2,T1]{fontenc}
\DeclareSymbolFont{cyrletters}{OT2}{wncyr}{m}{n} 
\DeclareMathSymbol{\Sha}{\mathalpha}{cyrletters}{"58}

\newtheorem{theorem}{Theorem}
\newtheorem{lemma}[theorem]{Lemma}
\newtheorem{corollary}[theorem]{Corollary}

\newcommand{\ad}{\mathrm{ad}}

\renewcommand{\acomm}{\alpha_{\mathrm{comm}}}

\renewcommand{\abs}[1]{\left\lvert#1\right\rvert}

\title{All You Need is Trotter}
\author[1]{Gumaro Rend\'on}
\email{grendon@fujitsu.com}
\affil[1]{Fujitsu Research of America, Inc, Santa Clara, CA 95054, USA}

\begin{document}

\begin{abstract}
    The work here enables a cost scaling of $\tilde{O}\left( \frac{T ^{1+o(1)}\,{\rm log}^2 (1/\epsilon)}{\varepsilon}\right)$ and maximum depth $\tilde{O}\left( \frac{T^{1+o(1)}\,{\rm log} (1/\epsilon)}{\varepsilon}\right)$ for estimating observables after some Hamiltonian evolution $e^{i H T}$ using no extra block-encoding qubits, where $\epsilon$ is the algorithmic error of the target Hamiltonian evolution and $\varepsilon$ is the statistical error from the observable estimation. Throughout, $\tilde{O}$ suppresses factors that are doubly logarithmic in $1/\epsilon$ (i.e. ${\rm polyloglog}(1/\epsilon)$), which are subdominant to the displayed ${\rm polylog}(1/\epsilon)$ factors. This is achieved through product formulas, stable interpolation (Chebyshev), and to calculate the needed fractional queries, cardinal sine interpolation is used. This is an improvement over similar work using interpolation methods which have a superlinear scaling of the form $5^p T^{1+1/p}$, where $p$ is the order of the product formula used and an improvement on required qubits over Quantum Signal Processing methods while maintaining scaling with ${\rm polylog}{1/\epsilon}$. This is accomplished through an optimized and variable (with $T$) order $p$. The slightly better $\log^2(1/\epsilon)$ compared to $\log^3(1/\epsilon)$ comes from the use of cardinal sine interpolation compared to other methods using Chebyshev interpolation and product formulas.
The work presented here also provides an alternative to Quantum Signal Processing techniques to calculate observables or amplitudes after a fractional query evolution with a cost $O(T\log{1/\epsilon})$ through this cardinal sine interpolation.
\end{abstract}

\maketitle

\section{Introduction}

A central task in quantum simulation is the estimation of observables following Hamiltonian evolution: given a Hamiltonian $H$, an evolution time $T$, and an observable or amplitude of interest, one wishes to estimate a quantity such as $\langle \psi_1 \vert e^{i H T} \vert \psi_2 \rangle$ to a target accuracy. Product formulas (Trotterization) are the most hardware-frugal route to this task, using no block-encoding ancillas, but they suffer from a poor dependence on the algorithmic error $\epsilon$: the Trotter step size must shrink polynomially in $1/\epsilon$. Post-Trotter methods such as Quantum Signal Processing achieve the desirable ${\rm polylog}(1/\epsilon)$ dependence, but at the cost of additional ancilla qubits and circuit-depth overhead from block-encoding. The question we address is whether one can retain the low qubit footprint of product formulas while attaining ${\rm polylog}(1/\epsilon)$ accuracy and a near-linear dependence on $T$.

This work answers the question in the affirmative. We give an algorithm for estimating observables after Hamiltonian evolution with a cost scaling of $\tilde{O}\!\left( T^{1+o(1)}\,\log^2 (1/\epsilon)/\varepsilon\right)$ and maximum depth $\tilde{O}\!\left( T^{1+o(1)}\,\log (1/\epsilon)/\varepsilon\right)$, using no extra block-encoding qubits, where $\epsilon$ is the algorithmic error and $\varepsilon$ is the statistical error of the observable estimation; here and throughout, $\tilde{O}$ hides factors doubly logarithmic in $1/\epsilon$, which are subdominant to the displayed polylogarithmic ones. Informally, the algorithm proceeds in three steps. (i) The Trotterized evolution is run at several step sizes chosen at stable (Chebyshev) interpolation nodes, and the corresponding amplitudes are estimated. (ii) Because these nodes require non-integer powers of the product formula, the needed fractional queries are reconstructed \emph{classically}, as a post-processing step, from the amplitudes of nearby integer-power product formulas via cardinal sine interpolation. (iii) The amplitudes are extrapolated in the Trotter step size to the zero-step-size limit, recovering the un-Trotterized result. The two design choices that yield the stated scaling are the use of cardinal sine interpolation for the fractional queries (which gives a $\log^2(1/\epsilon)$ dependence, an improvement over the $\log^3(1/\epsilon)$ of product-formula-plus-Chebyshev-interpolation methods) and an order $p$ that is optimized and allowed to vary with $T$, which produces the quasi-linear $T^{1+o(1)}$ dependence. As a byproduct, step (ii) is itself a self-contained method for evaluating amplitudes after a fractional-query evolution, at a cost $O(T\log{1/\epsilon})$ and without any signal processing on the quantum computer.

The remainder of this introduction places these ideas in the context of prior work on fractional queries and Trotter-error mitigation.

Fractional queries on a quantum computer are of great importance for quantum algorithms like Grover search and its variations. And they were difficult to come by up to 2006 according to Scott Aaronson's ``The ten most annoying questions in quantum computing"\cite{aaronson2006TenMostAnnoying}. They were later proven to be implementable by Sheridan et al. in ~\cite{Sheridan_fractional}. Later, an exponential improvement was achieved in the query model in the work of~\cite{Cleve_2009}.
These techniques combined with linear combination of unitaries and amplitude amplification authors of  ~\cite{Berry_Hamiltonian_sim_2013}  achieve an exponential improvement on Hamiltonian simulation.

After this, there was a flurry of works using Quantum Signal Processing~\cite{Low_2017,low2018hamiltonian} which achieved exponential precision improvements on several operations like inversion of a matrix, Hamiltonian simulation, and also fractional queries ~\cite{gilyen_SVT}. There has been a recent paradigm shift in regards to what has to be done in a quantum computer and what can be relegated to classical post-processing~\cite{Rendon_extrapol} using the power of stable interpolation. Here, it is shown how to implement fractional queries by using cardinal sine interpolation which is closely related to the sampling theorem \cite{whittaker1915xviii,kotel2006transmission,nyquist1928certain,shannon1948mathematical}, while preserving optimal error scaling and requiring no signal processing on the quantum computer.

Moreover, fractional queries have emerged as a necessity when using well-conditioned interpolation methods for interpolating to zero Trotter-step size\cite{Rendon_extrapol} of the Hamiltonian evolution operator as feasible alternative to block-encoding. These methods are competitive in the sense that they accomplish ${\rm polylog(1/\epsilon)}$ scaling, where $\epsilon$ is the algorithmic error. The scaling achieved for the evolution operator with respect to the time evolution $t$ was $O(t^2)$, and thus here an attempt to improve this is made. 

It is useful to situate the present work among existing approaches to Hamiltonian simulation and observable estimation. Product formulas\cite{Suzuki1991GeneralTO} are simple and use no ancilla qubits, but a $p$-order formula incurs an algorithmic error scaling as $T(T/r)^p$, so reaching error $\epsilon$ requires a step count growing polynomially in $1/\epsilon$. Post-Trotter methods --- linear combination of unitaries\cite{Berry_Hamiltonian_sim_2013}, quantum signal processing and qubitization\cite{Low_2017,low2018hamiltonian,Low2019hamiltonian}, and the quantum singular value transformation\cite{gilyen_SVT} --- achieve the exponentially better ${\rm polylog}(1/\epsilon)$ dependence and (for QSP/qubitization) optimal joint scaling in $T$ and $\epsilon$, but they require block-encoding the Hamiltonian, which introduces ancilla qubits and additional circuit depth. When the goal is the estimation of time-evolved observables rather than the evolution operator itself, a complementary line of work mitigates the Trotter error in classical post-processing: Richardson extrapolation and polynomial (Chebyshev) interpolation of observables computed at several Trotter step sizes\cite{Rendon_extrapol,Rendon2021,watson2024exponentially}, and multi-product formulas\cite{low2019well,watkins2024time}. These retain the low qubit footprint of product formulas while recovering a ${\rm polylog}(1/\epsilon)$ accuracy dependence; the underlying amplitudes are read out by amplitude estimation\cite{2000quant.ph..5055B,IQAE}. The algorithm presented here belongs to this last family. Relative to the Chebyshev-interpolation methods\cite{Rendon_extrapol,Rendon2021} on which it builds, it improves the accuracy dependence from $\log^3(1/\epsilon)$ to $\log^2(1/\epsilon)$ through cardinal sine interpolation of the fractional queries, and improves the dependence on $T$ to quasi-linear $T^{1+o(1)}$ through an order $p$ that varies with $T$; relative to the multi-product approach of \cite{low2019well}, discussed next, it forms the required combination classically rather than coherently, avoiding the associated qubit overhead.

We note that the well-conditioned multi-product formulas of \cite{low2019well} also confront the non-integer powers that arise from interpolation nodes, and do so by rounding them to integers while preserving a well-conditioned linear combination. The key distinction of the present work is where and how this combination is formed. In \cite{low2019well} the linear combination of (integer-power) product formulas is realized \emph{coherently and in situ} on the quantum device, through a linear combination of unitaries together with oblivious amplitude amplification; the rounding to integer powers is precisely what keeps this on-device combination well-conditioned. By contrast, the fractional queries here are never formed coherently: they are reconstructed \emph{classically, as a post-processing step}, from the amplitudes of integer-power product formulas via cardinal sine interpolation. Consequently, the Trotterization bias is corrected without the ancilla register and amplitude-amplification overhead required by an in-situ combination, that is, with no extra block-encoding qubits beyond those used for amplitude estimation.

In the next section (see \cref{sec:setup}), we illustrate how the necessity of fractional queries emerge in the approach of interpolating away Trotter error in a stable manner. In the section after (\cref{sec:fractional_queries}), we show how to obtain fractional queries through a classical addition of amplitudes of integer queries as a post-processing step using cardinal sine interpolation. Next (\cref{sec:cheb_interp}), bounds are obtained on the convergence rate of Chebyshev interpolation of the Trotterized evolution operator assuming access to efficient fractional queries. Finally (\cref{sec:combined_results}), the results are combined to obtain an algorithm that achieves quasi-linear cost-scaling with $t$ while keeping ${\rm polylog}(1/\epsilon)$.

The stability analysis for the cardinal sine interpolation is left to \Cref{app:sinc_stability} and the reader is pointed to \cite{Rendon_extrapol} for an analysis on the stability of Chebyshev interpolation.

\section{Setup of the problem\label{sec:setup}}
Suppose one wants to estimate:
\begin{align}
    \langle \psi_1 \vert \prod^{M}_{j} V_j e^{i T_j H_{j} } \vert \psi_2 \rangle,
\end{align}
where $V_i$ are unitary matrices.
Suppose each Hamiltonian $H_j$ on a collection of $n$ qubits can be decomposed into a sum of $m_j$ terms.
\begin{align}
    H_j = \sum^{m_j}_{\gamma=1} H_{j,\gamma}.
\end{align}
Notice that both the Hamiltonian and decomposition are specified, as the decomposition is not unique. Throughout this work, $\sum_\gamma^{m_j}\|H_{j,\gamma}\|\leq 1$ is assumed.

In order to implement the unitary evolution operators
\begin{align}
    e^{i T_j H_{j} }
\end{align}
we wish to approximate them first through $p$-order product formulas.
Standard examples of product formulas include the Trotter first-order formula
\begin{align}
S_{1,j} (t) := \prod^{m_j}_{\gamma=1} e^{i H_{j,\gamma} t},
\end{align}
the second-order symmetric formula
\begin{align}
S_{2,j} (t) :=  e^{i  H_{j,1} t/2} \dots  e^{i  H_{j,m_j} t/2} e^{i H_{j,m-j} t/2} \dots  e^{i  H_{j,1} t/2},
\end{align}
and more generally, the order $2k$ symmetric Suzuki-Trotter formula, defined recursively as
\begin{align} \label{eq:symm_S2k}
S_{2k,j} (t) := [S_{2k-2,j}(u_k t)]^2 S_{2k-2,j}\left((1 - 4 u_k) t\right) [S_{2k-2,j}(u_k t)]^2
\end{align}
for every $k\in \mathbb{Z}_+\setminus\{1\}$, where $u_k := (4-4^{1/(2k-1)})^{-1}$~\cite{Suzuki1991GeneralTO}.

Thus, one calculates
\begin{align}
    \langle \psi_1 \vert \prod^{M}_{j} V_j S^{g/s}_{p,j}(s t_{j})  \vert \psi_2 \rangle = \langle \psi_1 \vert \prod^{M}_{j} V_j \left(e^{i(s t_j)\tilde{H}_{p,j}(s t_j) } \right)^{(g/s)}  \vert \psi_2 \rangle,
\end{align}
where $g=T_j/t_j$, for certain values of $s$ and then interpolates to $s\to0$ to recover the desired result. Here,
\begin{align}
    \tilde{H}_{p,j}(s t_j) = \frac{\log\left(S_{p,j}(s t_j)\right)}{i s t_j },
\end{align}
where $\tilde{H}_{p,j}(s t_j)$ will be a single-valued analytic function of $s$ provided that
\begin{align}
      t_j   \leq \pi.
\end{align}
The values of $s$ chosen are those at optimal collocation of nodes, $s_k$, and set of interpolating polynomials such that the error does not blow up (see \cite{Rendon_extrapol,Gautschi90how(un)stable,Kuian_cheb,low2019well}). This, however, introduces the necessity of using fractional queries since $g/s_k$ will not be in general an integer. This was an obstacle which authors in \cite{Rendon2021} ran into. Typically, fractional queries are assumed to be calculated through Quantum Signal Processing techniques (\cite{gilyen_SVT,Low_2017,Low2019hamiltonian}) because this is the most efficient approach in terms of error. However, fractional queries using this technique require extra ancillary qubits and additional circuit depth overhead. For that reason, we propose to use close-to-optimal interpolation techniques by estimating nearby integer powers 
\begin{align}
    \langle \psi_1 \vert \prod^{M}_{j} V_j S^{m + n}_{p,j}(s_k t_{j})  \vert \psi_2 \rangle,
\end{align}
where $m=\lfloor g/s_k\rfloor$ for $n\in \{-q,-q+1,\dots,q\}$.

In the next section, we will detail how to obtain the fractional queries through cardinal sine interpolation and estimate upper-bounds on the resource requirements. Following that, we will estimate the uncertainty propagation to the interpolant due to cardinal sine interpolation. Finally, we explain the Trotter interpolation to the limit $s\to 0$.

\section{Fractional Queries using Cardinal Sine Interpolation\label{sec:fractional_queries}}

In this section, we will explain how one can use cardinal sine interpolation to obtain the fractional queries required by Trotter step interpolation. Here is the main theorem regarding fractional queries interpolation:

\begin{theorem}\label{thm:q}
In order to give an estimate of $\langle \psi_1 \vert \prod^{M}_{j} V_j S^{g/s_k}_{p,j}(s_k t_{j})  \vert \psi_2 \rangle$ within an error $\tilde\varepsilon_{\rm interp}$ using
\begin{align*}
     \frac{\sum^{q}_{n=-q}  \sinc\left(r-n\right) w(n)  \langle \psi_1 \vert \prod^{M}_{j} V_j S^{n+m}_{p,j}(s_k t_{j})  \vert \psi_2 \rangle }{w(r)},
\end{align*}
where $w(t) = \frac{1}{\sigma \sqrt{2\pi}}\exp\left(-\frac{t^2}{2{\sigma^2}}\right)$, $m=\lfloor g/s_k \rfloor$, and $r=g/s_k - \lfloor g/s_k \rfloor$, one needs to choose any
\begin{align*}
     q &\geq  \frac{6}{\pi}\left(\log\frac{ 4\sqrt{2\pi} + 8}{ \pi e^{1/3+\pi/12}  \tilde\varepsilon_{\rm interp}} \right)
\end{align*}
and enforce $\sum^{M}_{j=1} |t_j| \leq \pi $.
    
\end{theorem}

The proof occupies the remainder of this section and proceeds in three stages. We first cast the integer-power amplitudes as samples of a windowed signal $x(t)$ (\Cref{eq:xtilde}), so that recovering the fractional query becomes a band-limited interpolation problem. The interpolation error then splits, by the triangle inequality, into a \emph{spectral} part (from windowing a signal that is not exactly band-limited) and a \emph{truncation} part (from using finitely many samples). We bound these two contributions separately in \Cref{lem:spectral} and \Cref{lem:truncation}, and finally combine them to solve for the number of samples $q$ stated above.

To prove this theorem, we first make the following definition for the signal:
\begin{align}\label{eq:xtilde}
    \tilde{x}(t) &=\langle \psi_1 \vert \prod^{M}_{j} V_j S^{t+m}_{p,j}(s_k t_{j})  \vert \psi_2 \rangle = \langle \psi_1 \vert \prod^{M}_{j} V_j \left(e^{i(s_k t_j)\tilde{H}_{p,j}(s_k t_j) } \right)^{t+m}  \vert \psi_2 \rangle.
\end{align}
Now, one defines the following windowed function:
\begin{align}
    x(t) = w(t) \tilde{x}(t),
\end{align}
where
\begin{align}
    w(t) = \frac{1}{\sigma\sqrt{2\pi}}\exp\left(-\frac{t^2}{2{\sigma^2}}\right).
\end{align}

By triangle inequality, one can estimate the total error by looking at both the truncation error and the spectral error separately
\begin{align}
    \varepsilon_{\rm interp} = \max_t \left| x(t) - x_{{\rm interp},2q+1}(t) \right| &\leq \overbrace{\max_t \left| x(t) - x_{{\rm interp},\infty}(t) \right|}^{\varepsilon_{\rm spectral}} \cr
    &+ \underbrace{\max_t \left| x_{{\rm interp},\infty}(t) - x_{{\rm interp},2q+1}(t) \right|}_{\varepsilon_{\rm trunc}}.
\end{align}
Here,
\begin{align}
    x_{{\rm interp}, \infty} (t) = \sum_{n=-\infty}^{\infty} x(n) \sinc\left(t-n\right),
\end{align}
and
\begin{align}
    x_{{\rm interp}, 2q+1} (t) = \sum_{n=-q}^{q} x(n) \sinc\left(t-n\right).
\end{align}

Let us first bound the spectral error, $\varepsilon_{\rm spectral}$. Our result is expressed in the following lemma:
\begin{lemma}\label{lem:spectral}
Let $\Delta_{\rm pad}$ denote the guard band between the maximum normalized frequency of the signal, $\max_{\vec l}|\mu_{\vec l}|$, and the Nyquist edge $1/2$, that is, $\Delta_{\rm pad} = 1/2 - \max_{\vec l}|\mu_{\vec l}|$. Provided that $\Delta_{\rm pad} = 1/4$ and $ \sigma \geq 1/(4\pi)$, then
\begin{align*}
    \varepsilon_{\rm spectral} \leq \left(\frac{8}{\sqrt{2\pi}}+2\right)\exp\left(-\frac{\pi^2\sigma^2}{2}\right).
\end{align*}
\end{lemma}

To prove this, we start by re-expressing the operators in \Cref{eq:xtilde} terms of the eigen-states $\vert j, l \rangle$, defined through $\tilde{H}_{p,j}(s_k t_j) \vert j, l \rangle = E^{(j)}_{l} \vert j, l \rangle$, thus
\begin{align}
    e^{i(s_k t_j)\tilde{H}_{p,j}(s_k t_j) }  = \sum_{l}\vert j, l \rangle \langle j, l \vert \exp\left( i (s_k t_j) E^{(j)}_{l}\right),
\end{align}
\begin{align}
    V_j = \sum_{l',l''} \vert j, l' \rangle  \langle j, l' \vert V_j \vert j, l'' \rangle  \langle j, l'' \vert.
\end{align}
 We insert these expressions into \Cref{eq:xtilde}  to obtain
\begin{align}
    \tilde{x}(t) &= \langle \psi_1 \vert \prod^{M}_{j} V_j \left(e^{i(s_k t_j)\tilde{H}_{p,j}(s_k t_j) } \right)^{t+m}  \vert \psi_2 \rangle \cr 
    &=  \langle \psi_1 \vert\left(\prod_j \sum_{l',l''} \vert j, l' \rangle  \langle j, l' \vert V_j \vert j, l'' \rangle  \langle j, l'' \vert\sum_{l}\vert j, l \rangle \langle j, l \vert \exp\left( i (t+m) (s_k t_j) E^{(j)}_{l} \right) \right) \vert \psi_2 \rangle \cr 
    &=  \langle \psi_1 \vert\left(\prod_j \sum_{l,l'} \vert j, l' \rangle  \langle j, l' \vert V_j \vert j, l \rangle  \langle j, l \vert \exp\left( i (t+m) (s_k t_j) E^{(j)}_{l}\right) \right) \vert \psi_2 \rangle.
\end{align}
Here, it will be more convenient to define separate indices $l'$ and $l$ for each $j$, that is, $l'_{j}$ and $l_j$. This way, we can write the previous expression the following way:
\begin{align}
    \tilde{x}(t) &=  \sum^{2^n-1}_{l'_1, l'_2,\dots,l'_{M}}\sum^{2^n-1}_{l_1, l_2,\dots,l_{M}} \langle \psi_1 \vert\left(\prod_j  \vert j, l'_j \rangle  \langle j, l'_j \vert V_j \vert j, l_j \rangle  \langle j, l_j \vert \exp\left( i (t+m) (s_k t_j) E^{(j)}_{l_j}\right) \right) \vert \psi_2 \rangle.
\end{align}
With this labeling we can now factorize the phase factor:
\begin{align}
    \tilde{x}(t) &=  \sum^{2^n-1}_{l'_1, l'_2,\dots,l'_{M}}\sum^{2^n-1}_{l_1, l_2,\dots,l_{M}} \exp\left( \sum_{j'}i (t+m) (s_k t_{j'}) E^{(j')}_{l_{j'}}\right)  \langle \psi_1 \vert\left(\prod_j  \vert j, l'_j \rangle  \langle j, l'_j \vert V_j \vert j, l_j \rangle  \langle j, l_j \vert \right) \vert \psi_2 \rangle,
\end{align}
as long as we introduce the new dummy index $j'$. We can now regroup it
\begin{align}
    \tilde{x}(t) &=  \sum^{2^n-1}_{l_1, l_2,\dots,l_{M}} \exp\left( \sum_{j'}i (t+m) (s_k t_{j'}) E^{(j')}_{l_{j'}}\right)  \sum^{2^n-1}_{l'_1, l'_2,\dots,l'_{M}}\langle \psi_1 \vert\left(\prod_j  \vert j, l'_j \rangle  \langle j, l'_j \vert V_j \vert j, l_j \rangle  \langle j, l_j \vert \right) \vert \psi_2 \rangle.
\end{align}
We now define
\begin{align}
    o_{\vec{l}} = \sum^{2^n-1}_{l'_1, l'_2,\dots,l'_{M}}\langle \psi_1 \vert\left(\prod_j  \vert j, l'_j \rangle  \langle j, l'_j \vert V_j \vert j, l_j \rangle  \langle j, l_j \vert \right) \vert \psi_2 \rangle,
\end{align}
and
\begin{align}
    \mu_{\vec{l}} = \frac{1}{2\pi} \sum^M_{j=1} (s_k t_j) E^{(j)}_{l_j},
\end{align}
such that
\begin{align}
    \tilde{x}(t) = \sum_{\vec{l}} o_{\vec{l}} \,\, \tilde{x}^{(\mu_{\vec{l}})} (t),
\end{align}
where $\tilde{x}^{(\mu_{\vec{l}})}(t)=e^{i2\pi (t+m)\mu_{\vec{l}}}$. A similar decomposition for the windowed function will also be useful:
\begin{align}
    x(t) = \sum_{\vec{l}} o_{\vec{l}} \,\,x^{(\mu_{\vec{l}})}(t),
\end{align}
where $x^{(\mu_{\vec{l}})}(t) = w(t) \tilde{x}^{(\mu_{\vec{l}})}(t)$, more explicitly, for any $\mu$, each individual windowed signal looks like:
\begin{align}
    x^{(\mu)}(t) = \frac{1}{\sigma\sqrt{2\pi}}\exp\left(-\frac{t^2}{2{\sigma^2}}\right) \exp\left(i2\pi t \mu \right) \exp\left(i2\pi m \mu \right).
\end{align}
We also have that
\begin{align}
    \sum_{\vec{l}} \left|o_{\vec{l}}\right| \leq 1,
\end{align}
thus, bounding the interpolation error for the worst-case $x^{(\mu)}(t)$ should also be an upper bound for the interpolation error for $x(n)$. Simply put,
\begin{align}
    \varepsilon_{\rm spectral}=\max_{t}\left|x (t)-x_{{\rm interp}, \infty} (t)\right|\leq \max_{\mu,t}\left|x^{(\mu)}(t)-x^{(\mu)}_{{\rm interp}, \infty} (t)\right|,
\end{align}
where
\begin{align}
    x^{(\mu)}_{{\rm interp}, \infty} (t) = \sum_{n=-\infty}^{\infty} x^{(\mu)}(n) \sinc\left(t-n\right).
\end{align}
Moreover, 
\begin{align}
    x^{(\mu)}_{{\rm interp}, \infty} (t) &= \mathcal{F}_f^{-1}\left\{\mathcal{F}_{t'}\left\{\sum^{\infty}_{n=-\infty} x^{(\mu)}(n) \sinc\left(t'- n \right)  \right\}(f)\right\}(t) \cr 
    &= \mathcal{F}_f^{-1}\left\{\sum^{\infty}_{n=-\infty} x^{(\mu)}(n) \mathcal{F}_{t'}\left\{\sinc\left(t'- n\right)  \right\}(f)\right\}(t) \cr
    &= \mathcal{F}_f^{-1}\left\{\sum^{\infty}_{n=-\infty} x^{(\mu)}(n) \rect\left(f \right) \exp\left(2\pi i f n \right) \right\}(t) \cr
    &= \mathcal{F}_f^{-1}\left\{\rect\left(f \right) X^{(\mu)}_{1}(f) \right\}(t).    
\end{align}
The discrete-time Fourier transform $X^{(\mu)}_{1}(f)$ can be written in terms of the Fourier transform
\begin{align}
    X^{(\mu)}(f) &= \mathcal{F}_t\left\{x^{(\mu)}(t)\right\}(f) \cr
    &= \frac{1}{\sigma_f \sqrt{2\pi}} \exp\left(- \frac12 (f-\mu)^2/\sigma_f^2\right),
\end{align}
where $\sigma_f = \frac{1}{4\pi \sigma}$, in the following way:
\begin{align}
    X^{(\mu)}_{1}(f)&=\sum^{\infty}_{n=-\infty} x^{(\mu)}(n )  \exp\left(2\pi i f n \right) \cr
    &=\int \Sha\left(t\right) x^{(\mu)}(t)  \exp\left(2\pi i f t\right) \mathrm{d} t \cr
    &=\mathcal{F} \left\{\Sha\left(t\right)\right\} \star \mathcal{F} \left\{ x^{(\mu)}(t) \right\} \cr
    &=\Sha\left(f\right) \star X^{(\mu)}(f)  \cr
    &=\sum_k X^{(\mu)}(f-k).
\end{align}
Here, $\Sha$ is the Dirac comb function
$$
\Sha(t) = \sum^{\infty}_{k=-\infty} \delta (t-k).
$$
Now, the error between $x_{{\rm interp}, \infty}(t)$ and $x(t)$ is
\begin{align}\label{eq:eps_spec}
    \varepsilon_{\rm spectral}&\leq\max_{\mu}\left|\mathcal{F}^{-1}\left\{X^{(\mu)}(f)\right\} - \mathcal{F}^{-1}\left\{\rect\left(\frac{f}{f_s}\right) X^{(\mu)}_s(f)\right\}\right| \cr
    &\leq \max_{\mu}\left(\int{\sum_{k\neq 1} |X^{(\mu)}(f-k f_s)|\rect(f/f_s)} \mathrm{d}f \right) \cr
    &+\max_{\mu}\left(\int{ |X^{(\mu)}(f)|(1-\rect(f/f_s)}) \mathrm{d}f \right).
\end{align}
If one assumes
\begin{align}
    \frac{1}{2\pi} \sum^M_{j=1} t_j \leq 1/4,
\end{align}
we then know that $|\mu_{\vec{l}}|\leq 1/4$. Thus, one can bound the error by letting $\mu$ be in the worst case scenarios defined by
\begin{align}
    |\mu| = 1/4.
\end{align}
This means that the term in the last line of \Cref{eq:eps_spec} is bounded as
\begin{align}
    &\int{ |X^{(1/4)}(f)|(1-\rect(f/f_s)}) \mathrm{d}f \cr
    &\leq 2 \int^{\infty}_{f=1/2} \frac{1}{\sigma_f \sqrt{2\pi}} \exp\left(- \frac12 (f-1/4)^2/\sigma_f^2\right)  \mathrm{d}f \cr 
    & = 2\erfc \left( \frac{1}{\sqrt{32} \sigma_f}\right) \cr 
    &\leq 2 \exp \left( -\frac{1}{ 32 \sigma_f^2} \right).
\end{align}
For the term in the second-to-last line of \Cref{eq:eps_spec}, the $\mu$ that maximizes it is $\mu=-1/4$, giving us:
\begin{align}
    &\int{\sum_{k\neq 1} |X^{(-1/4)}(f-k f_s)|\rect(f/f_s)} \mathrm{d}f \cr 
    &\leq\max_f\left(\sum_{k\neq 1} |X^{(-1/4)}(f-k f_s)|\right)\overbrace{\int{\rect(f/f_s)} \mathrm{d}f}^{1} \cr
    &=2\sum^{\infty}_{k = 1} \frac{1}{\sigma_f \sqrt{2\pi}} \exp\left(- \frac12 (1/2+1/4 - k )^2/\sigma_f^2\right) \cr 
    &=2\sum^{\infty}_{k = 0} \frac{1}{\sigma_f \sqrt{2\pi}} \exp\left(- \frac12 (k+1/4)^2/\sigma_f^2\right) \cr 
    &= 2\sum^{\infty}_{k = 0} \frac{1}{\sigma_f \sqrt{2\pi}} \exp\left(- \frac12 (k^2+\frac12 k + 1/2)/\sigma_f^2\right) \cr 
    &= \frac{2 \exp\left(- \frac12 ( 1/2)/\sigma_f^2\right)}{\sigma_f \sqrt{2\pi}}\sum^{\infty}_{k = 0}  \left(\exp\left(- \frac14 1/\sigma_f^2\right) \right)^{(2k^2+ k)}.
\end{align}
Now, if we assume $\exp\left(- \frac{3}{4\sigma_f^2}\right)\leq1$ we can state the following
\begin{align}
    &\int{\sum_{k\neq 1} |X^{(-1/4)}(f-k f_s)|\rect(f/f_s)} \mathrm{d}f \cr 
    &\leq \frac{2 \exp\left(- \frac12 ( 1/2)/\sigma_f^2\right)}{\sigma_f \sqrt{2\pi}}\sum^{\infty}_{k = 0}  \left(\exp\left(- \frac14 1/\sigma_f^2\right) \right)^{(2k^2+ k)} \cr    
    &\leq \frac{2 \exp\left(- \frac1{4\sigma_f^2}\right)}{\sigma_f \sqrt{2\pi}}\sum^{\infty}_{k = 0}  \left(\exp\left(- \frac{3}{4\sigma_f^2}\right) \right)^{k} \cr
    & = \frac{2 \exp\left(- \frac1{4\sigma_f^2}\right)}{\sigma_f \sqrt{2\pi}} \frac{1}{1-\exp\left(- \frac{3}{4\sigma_f^2}\right)}.
\end{align}
If one further assumes
\begin{align}
    \exp\left(-\frac{3}{4\sigma_f^2}\right) \leq 1/2,
\end{align}
then
\begin{align}
    &\int{\sum_{k\neq 1} |X^{(-1/4)}(f-k f_s)|\rect(f/f_s)} \mathrm{d}f \cr 
    & \leq \frac{4 \exp\left(- \frac1{4\sigma_f^2}\right)}{\sigma_f \sqrt{2\pi}}.
\end{align}
We note that 
\begin{align}
    \frac{ \exp\left(- \frac1{4\sigma_f^2}\right)}{\sigma_f } = 2\sigma_f^2\frac{\partial \exp\left(-\frac{1}{4\sigma_f^2}\right) }{\partial \sigma_f},
\end{align}
which leads us to introduce the following necessary lemma:
\begin{lemma}[Derivative bound]\label{lem:fnbound}
If the function $f(z)$ is analytic on the complex disc of radius r centered at $z_0$, then 
\begin{align}
|f^{(n)}(z-z_0)| \leq \frac{M n! 2^n}{r^n},
\end{align}
for $|z-z_0| \leq \frac{r}{2}$, where $M=\max_{|z'-z_0|\leq r} |f(z'-z_0)|$.
\end{lemma}
Now, using this last lemma for $z-z_0=\sigma_f$, $z_0=2\sigma_f>0$ and $r=2\sigma_f$ one obtains
\begin{align}
    M = \max_{|z'-z_0|\leq 2\sigma_f} \exp\left(-\frac{1}{4(z' - 2\sigma_f)^2}\right) = \exp\left(-\frac{1}{4(4\sigma_f - 2\sigma_f)^2}\right) = \exp\left(-\frac{1}{16\sigma_f ^2}\right).
\end{align}
With this, we obtain the following bound:
\begin{align}
     2\sigma_f^2\frac{\partial \exp\left(-\frac{1}{4\sigma_f^2}\right) }{\partial \sigma_f} \leq 2 \sigma_f^2 \frac{2 }{2\sigma_f} \exp\left(-\frac{1}{16\sigma_f^2}\right) = 2 \sigma_f  \exp\left(-\frac{1}{16\sigma_f^2}\right).
\end{align}
Finally, if we restrict ourselves to considering $\sigma_f\leq 1$,
\begin{align}
    \frac{ \exp\left(- \frac1{4\sigma_f^2}\right)}{\sigma_f } \leq 2 \exp\left(-\frac{1}{16\sigma_f^2}\right). 
\end{align}
Thus,
\begin{align}
    &\int{\sum_{k\neq 1} |X^{(-1/4)}(f-k f_s)|\rect(f/f_s)} \mathrm{d}f \cr 
    &\leq \frac{8}{\sqrt{2\pi}}\exp\left(-\frac{1}{16\sigma_f^2}\right) \leq \frac{8}{\sqrt{2\pi}}\exp\left(-\frac{1}{32\sigma_f^2}\right).
\end{align}
Putting together the bounds on the last two terms in \Cref{eq:eps_spec}, one gets
\begin{align}
    \varepsilon_{\rm spectral} \leq \left(\frac{8}{\sqrt{2\pi}}+2\right)\exp\left(-\frac{1}{32\sigma_f^2}\right). 
\end{align}
Thus, we have proven \Cref{lem:spectral}.

Now, we wish to obtain a bound on the truncation error, $\varepsilon_{\rm trunc}$. This is summarized in the following lemma:
\begin{lemma}\label{lem:truncation}
\begin{align*}
        \varepsilon_{\rm trunc} & =\max_{t}\left| x_{{\rm interp},\infty}(t) - x_{{\rm interp},2q+1}(t) \right| \leq 2 \exp\left( - \frac{(q+2)^2}{2\sigma^2} \right).
\end{align*}
\end{lemma}
\begin{proof}

\begin{align}
    \varepsilon_{\rm trunc} & =\max_{t}\left| x_{{\rm interp},\infty}(t) - x_{{\rm interp},2q+1}(t) \right|\cr 
&\leq\max_{t,\mu}\left| x^{(\mu)}_{{\rm interp},\infty}(t) - x^{(\mu)}_{{\rm interp},2q+1}(t) \right|\cr     
    &\leq\sum_{ n\not\in \{k\mid -q \leq  k \leq q\} }  \left|x^{(\mu)}(t)\sinc \left(t-n\right)\right| \cr
    &\leq\sum_{ n\not\in \{k\mid -q \leq  k \leq q\} }  \left\vert x^{(\mu)}(t)\right\vert \cr 
    &\leq\sum_{ n\not\in \{k\mid -q \leq  k \leq q\} }  \left\vert w(t)\right\vert \cr    
    & \leq\sum_{ n\not\in \{k\mid -q \leq  k \leq q\} }\frac{1}{\sigma \sqrt{2\pi}}\exp\left(-\frac{n^2}{2{\sigma^2}}\right) \cr 
    &=2\sum^{\infty}_{ n = q+1} \frac{1}{\sigma\sqrt{2\pi}}\exp\left(-\frac{n^2}{2{\sigma^2}}\right) \cr 
    &\leq 2\int^{\infty}_{ n = q+2} \frac{1}{\sigma\sqrt{2\pi}}\exp\left(-\frac{n^2}{2{\sigma^2}}\right) \cr 
    &= 2\erfc\left(\frac{(q+2)}{\sqrt{2}\sigma}\right)\leq 2 \exp\left( - \frac{(q+2)^2}{2\sigma^2} \right).
\end{align}
\end{proof}

Now, we set equal the exponentials on the bounds for $\varepsilon_{\rm spectral}$ and $\varepsilon_{\rm trunc}$ from \Cref{lem:spectral} \Cref{lem:truncation} respectively:
\begin{align}
    \exp\left( - \frac{(q+2)^2}{2\sigma^2} \right) = \exp\left(-\frac{\pi^2\sigma^2}{2}\right),
\end{align}
which gives
\begin{align}
    \sigma_f = \frac{1}{4\sqrt{\pi}}\frac{1}{\sqrt{q+2}}.
\end{align}
With this, $\varepsilon_{\rm interp}$ can be bounded through
\begin{align}
    \varepsilon_{\rm interp} \leq \left(4 + \frac{8}{\sqrt{2\pi}}\right)\exp\left(-\frac{\pi(q+2)}{2}\right).
\end{align}
Thus, solving for $q$, we obtain an upper bound for it:
\begin{align}
  q\leq \frac{2}{\pi}\log\frac{\left(4 + \frac{8}{\sqrt{2\pi}}\right)}{\varepsilon_{\rm interp}} -2.
\end{align}
The interpretation of this bound is the following: ``If we have an error $\varepsilon_{\rm interp}$, the value of $q$ must be at most $\frac{2}{\pi}\log\frac{\left(4 + \frac{8}{\sqrt{2\pi}}\right)}{\varepsilon_{\rm interp}} -2$''.

At this point, the error $\varepsilon_{\rm interp}$ is with respect to the windowed function
\begin{align}\label{eq:wtilde}
    \tilde w(t) := w(t)\langle \psi_1 \vert \prod^{M}_{j} V_j S^{t+m}_{p,j}(s_k t_{j})  \vert \psi_2 \rangle.
\end{align}
Now, the bounds on $q$ with respect to the error
\begin{align}
    \tilde\varepsilon_{\rm interp} = \left| \tilde w(r) - x_{{\rm interp},2q+1}(r) \right| = \varepsilon_{\rm interp}/w(r).
\end{align}
After making the substitution the upper bound for $q$ becomes:
\begin{align}
  q&\leq \frac{2}{\pi}\log\frac{\left(4 + \frac{8}{\sqrt{2\pi}}\right)}{w(r)\tilde\varepsilon_{\rm interp}} -2 \cr
       &\leq  \frac{2}{\pi}\log\frac{\sigma\left(4\sqrt{2\pi} + 8\right)}{\exp\left(-\frac{r^2}{2{\sigma^2}}\right)\tilde\varepsilon_{\rm interp}} -2 \cr
     &\leq  \frac{2}{\pi}\log\frac{\sigma\left(4\sqrt{2\pi} + 8\right)}{\tilde\varepsilon_{\rm interp}} + \frac{2}{\pi}\frac{1}{2{\sigma^2}} -2 \cr
    &\leq  \frac{2}{\pi}\log\frac{\left(4\sqrt{2\pi} + 8\right)}{\pi \sigma_f\tilde\varepsilon_{\rm interp}} + \frac{2}{\pi}\frac{16 \pi ^2 \sigma_f^2}{2} -2 \cr
  &\leq  \frac{2}{\pi}\log\frac{\sqrt{q+2}\left(4\sqrt{2\pi} + 8\right)}{\sqrt{\pi} \tilde\varepsilon_{\rm interp}} + \frac{2}{\pi}\frac{ \pi }{2(q+2)} -2 \cr
  &\leq  \frac{1}{\pi}\log\frac{(q+2)\left(4\sqrt{2\pi} + 8\right)^2}{ \pi \tilde\varepsilon_{\rm interp}^2} + \frac{1}{\pi}\frac{ \pi }{(q+2)} -2   
\end{align}
If we assume that $q\geq 0$, we can loosen the bound to
\begin{align}
     (q+2) &\leq  \frac{1}{\pi}\log\frac{(q+2) \exp\left(\pi/2\right)\left(4\sqrt{2\pi} + 8\right)^2}{ \pi \tilde\varepsilon_{\rm interp}^2}
\end{align}

Solving for $q$ one obtains
\begin{align}
     (q+2) &\leq  -\frac{1}{\pi}W_{-1}\left(-\frac{ \pi^2 \tilde\varepsilon_{\rm interp}^2}{ \exp\left(\pi/2\right)\left(4\sqrt{2\pi} + 8\right)^2}\right),
\end{align}
provided $0\leq \frac{ \pi^2 \tilde\varepsilon_{\rm interp}^2}{ \exp\left(\pi/2\right)\left(4\sqrt{2\pi} + 8\right)^2} \leq 1/e$. Now, we use $-W_{-1}(-e^{-u-1})\leq 1 +\sqrt{2 u}+u\leq 1+3 u$ for $u>0$ We are now able to identify
\begin{align}
    u = \log\frac{ \exp\left(\pi/2\right)\left(4\sqrt{2\pi} + 8\right)^2}{ e \pi^2 \tilde\varepsilon_{\rm interp}^2},
\end{align}
which means $\frac{ \pi^2 \tilde\varepsilon_{\rm interp}^2}{ \exp\left(\pi/2\right)\left(4\sqrt{2\pi} + 8\right)^2} < 1/e$. Using the aforementioned bound, one obtains
\begin{align}
     q &\leq  \frac{6}{\pi}\left(\log\frac{ 4\sqrt{2\pi} + 8}{ \pi e^{1/3+\pi/12}  \tilde\varepsilon_{\rm interp}} \right).
\end{align}
This means that, for a given target error $\tilde\varepsilon_{\rm interp}$ we need only choose any:
\begin{align}
     q &\geq  \frac{6}{\pi}\left(\log\frac{ 4\sqrt{2\pi} + 8}{ \pi e^{1/3+\pi/12}  \tilde\varepsilon_{\rm interp}} \right).
\end{align}
Thus, we have proven \Cref{thm:q}.

\section{Optimally-conditioned Trotter Interpolation \label{sec:cheb_interp}}

After having solved the issue interpolating to fractional queries with powers $g/s_k$ we will now discuss the interpolation of the amplitudes
\begin{align}
    f(s)=\langle \psi_1 \vert \prod^{M}_{j} V_j S^{g/s}_{p,j}(s t_{j})  \vert \psi_2 \rangle = \langle \psi_1 \vert \prod^{M}_{j} V_j \left(e^{i(s t_j)\tilde{H}_{p,j}(s t_j) } \right)^{(g/s)}  \vert \psi_2 \rangle
\end{align}
to $s\to 0$. These results will be summarized at the end of the section in \Cref{thm:Troter_Cheb_interp}.

The interpolant we use is the following linear combination of orthogonal polynomials
\begin{align}\label{eq:O_expans_opt}
   P_{n-1} f(s) = \sum_{j=0}^{n-1} c_{j} p_j (s).
\end{align}
One samples $f$ at certain values $s_k$ and solves the following system of linear equations
\begin{align} \label{eq:cheb_sys}
    y = \mathbf{V} c
\end{align}
where $y = (f(s_1), f(s_2), \dots, f(s_n))$, and
\begin{align}
 \mathbf{V} :=
\begin{pmatrix}
 p_0(s_1)   & p_1(s_1)   & \dots  & p_{n-1}(s_1) \\
 p_0(s_2)   & p_1(s_2)   & \dots  & p_{n-1}(s_2) \\
 \vdots     & \vdots     & \ddots & \vdots  \\
 p_0(s_{n}) & p_1(s_{n}) & \dots  & p_{n-1}(s_{n})
\end{pmatrix}.
\end{align}
Thus, one can obtain the coefficients, $c_j$, with $c=\mathbf{V}^{-1}y$. Like in \cite{Rendon_extrapol}, the choice of interpolating nodes $s_k$ and the interpolating set of polynomials are the Chebyshev nodes and polynomials.
In that case, $p_j$ is defined by
\begin{align} \label{eq:cheb_orthonorm_def}
     p_j(s) :=
    \begin{cases}
    \sqrt{\frac1n}T_0(s), &j=0 \\
    \sqrt{\frac2n}T_j(s), &j=1,2,\dots \\
    \end{cases} 
\end{align}
where $T_j$ is the standard $j$th Chebyshev polynomial.
\begin{align}
    T_j(x) := \cos (n \cos^{-1} x).
\end{align}
The node collocation is described by
\begin{align} \label{eq:cheb_node_def}
    s_i = \cos\left(\frac{2i-1}{2n} \pi\right).
\end{align}
These polynomials fulfill the discrete orthonormality condition~\cite{mason2002chebyshev} with respect to the collocation nodes, that is,
\begin{align}
    \sum_{k=1}^n p_i (s_k) p_j (s_k) = \delta_{ij}
\end{align}
for all $0\leq i, j < n$. Instead of estimating the coefficients $c_i$ intermediately to estimate $f(0)$, one knows from \cite{Rendon_extrapol}, that 
\begin{align}
   P_{n-1}f(0)=\sum^n_{k=1} d_k f(s_k),
\end{align}
where
\begin{align}
    d_k(0) &= \frac{1}{n}(-1)^{k+n/2} \tan\left(\frac{2k-1}{2n}\pi\right).
\end{align}
We define
\begin{align}
    \epsilon_{\rm cheb} = |f(0) - P_{n-1}f(0)|.
\end{align}
Now, to get estimates on the Trotter interpolation error bounds, we need to introduce some methods of complex analysis. For each $\rho > 1$, let $B_\rho \subset \mathbb{C}$ be the Bernstein ellipse, which is an ellipse with foci at $\pm 1$ and the semimajor axis is $(\rho + \rho^{-1})/2$. The following lemma bounds the Chebyshev interpolation error for analytic functions on $B_\rho.$
\begin{lemma}\label{lem:Berns}
    Let $f(z)\in \mathbb{C}$ be an analytic function on $B_\rho$, and suppose $C\in\mathbb{R}_+$ is an upper bound such that $|f(z)| \le C$, for all $z \in B_\rho$. Then the Chebyshev interpolation error on $[-1,1]$ satisfies
    $$
    \norm{f - P_n f}_{\infty} \leq \frac{4 C \rho^{-n}}{\rho - 1}
    $$
    for each degree $n>0$ of the interpolant through the $n+1$ Chebyshev nodes.
\end{lemma}
\begin{proof}
    Theorem 8.2 of Ref.~\cite{trefethen2019approximation}.
\end{proof}

With this lemma in place, we now wish to upper bound the corresponding norm of the amplitude $f(z)$ on the Bernstein ellipse by promoting the real variable $s$ to $z \in B_\rho$.

First, like in \Cref{sec:setup}, we assume a sufficiently small $|\tau_i|$, where $\tau_i = z t_j$, such that
\begin{align}
    \tilde{H}_j(\tau_j) = \frac{\log\left( S_{p,j}(\tau_j)\right)}{i\tau_j}
\end{align}
is single valued. Now, recalling, with slight modifications, Lemma 14 from~\cite{Rendon_extrapol}:
\begin{lemma}\label{lem:H_error_bound}
Let $\tilde{H}_j (\tau_j)$ be the effective Hamiltonian associated with a complex-time $p$th order product formula $S_{p,j}(\tau_j)$ , where $\tau_j = z t_j$ with $z\in\mathbb{C}$ and $t_j\in\mathbb{R}_+$, which is analytically continued to an open neighborhood containing $z\in [-1,1]$. The norm of operator error can upper bounded as
\begin{align*}
\norm{H_j - \tilde{H}_j} \leq \frac{5}{2}\sum_{(\upsilon,m)} \acomm\big(H_{j,\pi_{\Upsilon}(\Gamma)},\ldots,H_{j,\pi_{\upsilon}(m+1)},H_{j,\pi_{\upsilon}(m)}\big) \cr
\cdot\frac{\abs{\tau}^p}{(p+1)!} e^{2|\tau_j| \, \sum^{m}_{j} \left\| H_{j} \right\|}
\end{align*}
provided that $\norm{H - \tilde{H}_j}\leq15$, $|\tau|\leq 1/8$, and
\begin{align*}
\sum_{(\upsilon,m)} \acomm\big(H_{j,\pi_{\Upsilon}(\Gamma)},\ldots,H_{j,\pi_{\upsilon}(m+1)},H_{j,\pi_{\upsilon}(m)}\big)\frac{\abs{\tau}^p}{(p+1)!} e^{2|\tau_j| \, \sum^{}_{m_j} \left\| H_{j,m_j} \right\|} \leq 1.
\end{align*}
\end{lemma}
Here, $\acomm\big(A_s,\ldots,A_1,B\big):=\sum_{q_1+\cdots+q_s=p}\binom{p}{q_1\ \cdots\ q_s}\norm{\ad_{A_s}^{q_s}\cdots\ad_{A_1}^{q_1}(B)}$. This means that, for a $p$th-order product formula one has that for $L_j(\tau_j)=\tilde{H}_j(\tau_j)-H_j$
\begin{align}
    \|L_j(\tau_j)\| = O \left(\frac{|\tau_j|^p \alpha_j}{(p+1)!} \right),
\end{align}
provided the constraints of Lemma 14 of ~\cite{Rendon_extrapol}  are fulfilled. Here, 
\begin{align}
    \alpha_j = \sum_{(\upsilon,m)} \acomm\big(H_{j,\pi_{\Upsilon}(\Gamma)},\ldots,H_{j,\pi_{\upsilon}(m+1)},H_{j,\pi_{\upsilon}(m)}\big).
\end{align}
We specifically look at each evolution operator that one obtains from Trotterization with a step $\tau_j$ is
\begin{align}
    e^{i T_j(H_j+L_j(\tau_j))}.
\end{align}
We split this operator (or equivalently raise to a real fractional power $\Delta/T_j$) into sufficiently small steps of length $\Delta$, such that one can use the Zassenhauss formula. That is, we first look at the following expression and use the Zassenhauss formula:
\begin{align}
    e^{-i\Delta H_j}e^{i\Delta(H_j+L_j(\tau_j))} = e^{i\Delta L_j(\tau_j)} e^{  \Delta^2 [H_j,L_j(\tau_j)]}\dots
\end{align}
Thus, applying $e^{i \Delta H_j}$ on the left, we finally obtain:
\begin{align}
    e^{i\Delta(H_j+L_j(\tau_j))} = e^{i\Delta H_j}e^{i\Delta L_j(\tau_j)} e^{  \Delta^2 [H_j,L_j(\tau_j)]}\dots.
\end{align}
With this, one has that 
\begin{align}
    \|e^{i\Delta(H_j+L_j(\tau_j))}\| \leq \left\| e^{i \Delta H_j}\right\|  O\left(e^{\Delta|\tau_j|^p\alpha_j/(p+1)!}\right) =  O\left(e^{\Delta|\tau_j|^p\alpha_j/(p+1)!}\right),
\end{align}
which finally gives us, for the full evolution length $T_j$, after exponentiating $e^{i\Delta(H_j+L_j(\tau_j))}$ by $T_j/\Delta$, a real number, we can bound the following norm through:
\begin{align}
   \left\| e^{iT_j (H_j+L_j(\tau_j))}\right\| = O\left(e^{ T_j|\tau_j|^p \alpha_j/(p+1)!}\right).
\end{align}
This means that,
\begin{align}
    \max_{z\in B_{\rho}}\left|\langle \psi_1 \vert \prod^{M}_{i} V_i S^{g/z}_{p}(z t_{i})  \vert \psi_2 \rangle \right| = O\left(e^{ \frac{g}{(p+1)!}\sum_j t_j \alpha_j(r t_j)^p }\right)
\end{align}
with $g=T_j/t_j$. Here, $r=O(1)$ is the radius of the smallest disc containing $B_{\rho}$. The relationship between $\rho$ and $r$ is
\begin{align} 
  \rho = r + \sqrt{r^2-1}.
\end{align}
Now, using \Cref{lem:Berns}, one knows the interpolation error is
\begin{align}
    \epsilon_{\rm cheb} &= O\left(\frac{e^{\frac{g}{(p+1)!}\sum_j \alpha_j t_j (r t_j)^p}}{r^n}\right) \cr 
    &= O\left(\frac{e^{\frac{g}{(p+1)!}\sum_j \alpha_j t_j (r t_j )^p}}{e^{n\log{r}}}\right).
\end{align}
Solving for $n$ one gets
\begin{align}
   n=\frac{1}{\log{r}} O\left(\log{1/\epsilon_{\rm cheb}} + \frac{g}{(p+1)!}\sum_j\alpha_j t_j (r t_j )^p \right).
\end{align}

This is summarized in the following theorem:

\begin{theorem}\label{thm:Troter_Cheb_interp}
In order to get a error
\begin{align}
    \epsilon_{\rm cheb} = \left|\langle \psi_1 \vert \prod^{M}_{j} V_j e^{i T_j H_{j} } \vert \psi_2 \rangle - P_{n-1} f(0)\right|,
\end{align}
where $f(s)=\langle \psi_1 \vert \prod^{M}_{j} V_j S^{g/s}_{p,j}(s t_{j})  \vert \psi_2 \rangle$, $g=T_j/t_j$, and
\begin{align}
   P_{n-1}f(0)=\sum^n_{k=1} d_k f(s_k),
\end{align}
where $d_k(0) = \frac{1}{n}(-1)^{k+n/2} \tan\left(\frac{2k-1}{2n}\pi\right)$, and $s_i = \cos\left(\frac{2i-1}{2n} \pi\right)$, one needs
\begin{align*}
   n=\frac{1}{\log{r}} O\left(\log{1/\epsilon_{\rm cheb}} + \frac{g}{(p+1)!}\sum_j\alpha_j t_j (r t_j )^p \right),
\end{align*}
where
\begin{align*}
    \alpha_j = \sum_{(\upsilon,m)} \acomm\big(H_{j,\pi_{\Upsilon}(\Gamma)},\ldots,H_{j,\pi_{\upsilon}(m+1)},H_{j,\pi_{\upsilon}(m)}\big),
\end{align*}
and $m_j$ runs over all the non-commuting terms of the product-formula decomposition for $H_j$, and provided that
\begin{align}
     t_j \leq \pi.
\end{align}
\end{theorem}

\section{Combined Results\label{sec:combined_results}}

In this work, a distinction between algorithmic and statistical error is made. We call statistical error, $\varepsilon$, the error that comes from estimating amplitudes on the quantum computer. We call algorithmic error, $\epsilon$, the error stemming from the use of finite interpolations, assuming that the underlying amplitudes are estimated exactly. 
In turn, from triangle inequality, we know that the algorithmic error is
\begin{align}
    \epsilon \leq \epsilon_{\rm cheb} + \epsilon_{\rm sinc}.
\end{align}
For convenience, we set $\epsilon_{\rm cheb} = \epsilon_{\rm sinc} = \epsilon/2$. Now, we bring to your attention that the number of queries required for Trotter interpolation is proportional to
\begin{align}
    \text{total cost} =O\left(\frac{\log{q} \log{n}}{\varepsilon} 5^{p} M\sum^n_i q \frac{g}{|s_i|}\right),
\end{align}
and a maximum circuit depth that goes like:
\begin{align}
  \text{max depth} =   O\left(\frac{\log{q} \log{n}}{\varepsilon} 5^{p} M  \frac{g}{\min_i|s_i|}\right),
\end{align}
where $s_i$ is defined in \cref{eq:cheb_node_def}, and the factor $5^p$ is coming from the number of stages in the Suzuki product formulas. The $1/\varepsilon$ factor is on account of the amplitude estimation algorithm needed to extract the underlying amplitudes \cite{2000quant.ph..5055B,IQAE}. One need only estimate the underlying amplitudes with statistical error ${O}\left(\frac{\varepsilon}{\log{q} \log{n}}\right)$ to get a statistical error $\varepsilon$ on the final interpolant. This is due to the stability of the cardinal sine interpolation (see \Cref{app:sinc_stability}) and the stability of Chebyshev interpolation (see lemmas 4 and 5 of \cite{Rendon2021} ). The factor $g/|s_i|$ is the number of queries needed to achieve an evolution time  $O\left(T_j\right)$ for each Chebyshev sample point $s_i$, and $q$ are the number of integer queries around $\lceil g / |s_i| \rceil$ needed to interpolate to an evolution time $T_i$. 
From \cite{Rendon_extrapol}, one knows that $\sum_i 1/|s_i|=O(n\log{n})$, and $1/\min_i|s_i|= O(n)$, thus, finally, the number of  queries is proportional to
\begin{align}
   \text{total cost} = O\left( 5^{p} M g n\frac{q\log q \log^2{n}}{\varepsilon}\right),
\end{align}
and the maximum circuit depth is 
\begin{align}
\text{max depth} = O\left(5^p M g n\frac{ \log q \log n }{\varepsilon}\right).
\end{align}
From these we take the factor $5^p M g n$ and define the following cost function:
\begin{align}
    L = 5^p M g n.
\end{align}
Replacing $n$ from \Cref{thm:Troter_Cheb_interp} we get
\begin{align}
    L=O\left( 5^{p} M g \frac{1}{\log{r}} \left(\log{1/\epsilon_{\rm cheb}} + \frac{g}{(p+1)!}\sum_j\alpha_j t_j (r t_j )^p \right) \right).
\end{align}
We can simplify this using:
\begin{align}
    t_j  = T_j/g \leq \frac{\max_j{T_j}}{g},
\end{align}
we define $T_{\rm max} = \max_j T_j$ and $\alpha_{
\rm max}=\max_j \alpha_j$, and thus write the  cost function as:
\begin{align}
  L = O\left( 5^{p} M g \frac{1}{\log{r}} \left(\log{1/\epsilon_{\rm cheb}} + \frac{M \alpha_{\rm max} T_{\rm max} ( r T_{\rm max}/g )^p}{(p+1)!} \right) \right).
\end{align}
For a moment, we treat $p$ as a continuous variable and optimize this bound with respect to $g$ to obtain:
\begin{align}
g_{\rm opt} = r\left( \frac{(p-1) M \alpha_{\text{max}} T_{\text{max}}^{p+1} }{(p+1)!\, \log (1/\epsilon_{\text{cheb}})} \right)^{\frac{1}{p}}.
\end{align}
However, for simplicity, we choose a close optimum
\begin{align}
    \tilde{g}_{\rm opt} \sim r\left( \frac{ M \alpha_{\text{max}} T_{\text{max}}^{p+1} }{(p+1)!} \right)^{\frac{1}{p}}
\end{align}
If we replace this back into the cost function, we get:
\begin{align}
    L = O  \left(\frac{5^p (M T_{\rm max})^{1+1/p} r}{\log r} \left(\log \frac{1}{\epsilon_{\text{cheb}}}\right) \left(   \alpha_{\text{max}}   \right)^{\frac{1}{p}} \right).
\end{align}
Here, we now assume we use any $r>1$ for which $r/\log{r} = O(1)$. Here, we will now optimize the cost with respect to the order, $p$. If we optimize this expression with respect to $p$, we obtain:
\begin{align}
    p_{\rm opt} =\sqrt{\frac{\log{M T_{\rm max} \alpha_{\rm max} }}{\log{5}}}.
\end{align}
However, to avoid any problems for when $\alpha_{\rm max} \ll 1$ , we simply choose a near optimal scaling for $p$ that scales only with $M T_{\rm max}$ (or optimize again, but without the factor $\alpha_{\rm max}^{1/p}$):
\begin{align}
    \tilde{p}_{\rm opt} = \sqrt{\frac{\log{M T_{\rm max}}}{\log 5}}.
\end{align}
Replacing this choice for $p$ in the bound for the cost function $L$, we obtain:
\begin{align}
    L = O  \left({ (M T_{\rm max})^{1+o(1)}}  \alpha_{\rm max}^{ \Theta\left(\sqrt{\frac{1}{\log\left( M T_{\rm max} \right)}} \right)}  e^{ \Theta\left(\sqrt{\log\left(M T_{\rm max} \right) } \right)}  \log \frac{1}{\epsilon_{\text{cheb}}} \right).
\end{align}
We replace this cost function bound in the previous total gate cost bounds and maximum circuit depth:
\begin{align}
    \text{total cost} =  O  \left({ (M T_{\rm max})^{1+o(1)}}  \alpha_{\rm max}^{ \Theta\left(\sqrt{\frac{1}{\log\left( M T_{\rm max} \right)}} \right)}  e^{ \Theta\left(\sqrt{\log\left(M T_{\rm max} \right) } \right)}  \log \frac{1}{\epsilon_{\text{cheb}}} \frac{q\log q \log^2{n}}{\varepsilon} \right),\cr 
    \text{max depth} =  O  \left({ (M T_{\rm max})^{1+o(1)}}  \alpha_{\rm max}^{ \Theta\left(\sqrt{\frac{1}{\log\left( M T_{\rm max} \right)}} \right)}  e^{ \Theta\left(\sqrt{\log\left(M T_{\rm max} \right) } \right)}  \log \frac{1}{\epsilon_{\text{cheb}}} \frac{\log q \log{n}}{\varepsilon} \right).
\end{align}
We recall that $q=O(\log \left(1/\epsilon_{\rm sinc}\right) )$ from \cref{thm:q} and note that $n=O(\log(1/\epsilon_{\rm cheb}))$ after our choice of $g$ and $p$. If we set both target error tolerances equal: $\epsilon_{\rm sinc}=\epsilon_{\rm cheb} = \epsilon/2$ , we have:
\begin{align}
   \text{total cost} = \tilde{O}  \left( \frac{ (M T_{\rm max})^{1+o(1)} \alpha_{\rm max}^{ \Theta\left(\sqrt{\frac{1}{\log\left( M T_{\rm max} \right)}} \right)}}{\varepsilon} \log^2 \left(\frac{1}{\epsilon} \right) \right)\cr 
     \text{max depth} = \tilde{O}  \left( \frac{ (M T_{\rm max})^{1+o(1)} \alpha_{\rm max}^{ \Theta\left(\sqrt{\frac{1}{\log\left( M T_{\rm max} \right)}} \right)}}{\varepsilon} \log \left(\frac{1}{\epsilon} \right) \right),
\end{align}
where, by triangle inequality, the algorithmic error is going to be less than the tolerance $\epsilon$. We have also grouped the term $e^{ \Theta\left(\sqrt{\log\left(M T_{\rm max} \right) } \right)}$ into $(MT_{\rm max})^{1+o(1)}$ since the exponent $o(1)$ already encompasses all sub-polynomial growth terms.  We do something similar to the bound on the exponent of $\alpha_{\rm max}$ and replace it with $o(1)$:
\begin{align}
    \text{total cost} = \tilde{O}  \left( \frac{ (M T_{\rm max})^{1+o(1)} \alpha_{\rm max}^{ o(1)}}{\varepsilon} \log^2 \left(\frac{1}{\epsilon} \right) \right), \cr 
    \text{max depth} = \tilde{O}  \left( \frac{ (M T_{\rm max})^{1+o(1)} \alpha_{\rm max}^{ o(1)}}{\varepsilon} \log \left(\frac{1}{\epsilon} \right) \right).
\end{align}
We also note that with our choices for $g$ and $p$, the number of Chebyshev nodes for interpolation are:
\begin{align}
   n=O\left(\log{1/\epsilon}  \right).
\end{align}
We can summarize these results in the following theorem:

\begin{theorem}\label{thm:Troter_Cheb_interp_altogether}
We can estimate
\begin{align*}
    \langle \psi_1 \vert \prod^{M}_{j} V_j e^{i T_j H_{j} } \vert \psi_2 \rangle,
\end{align*}
with an algorithmic error of at most $\epsilon$ and a statistical error of at most $\varepsilon$, with a query cost of
\begin{align*}
    \tilde{O} \left(  \frac{ (M T_{\rm max})^{1+o(1)}\alpha_{\rm max}^{o(1)}}{\varepsilon} \log^2 \frac{1}{\epsilon} \right) ,
\end{align*}
and maximum circuit depth going like:
\begin{align*}
   \tilde{O}  \left( \frac{ (M T_{\rm max})^{1+o(1)}\alpha_{\rm max}^{o(1)}}{\varepsilon} \log \frac{1}{\epsilon} \right),
\end{align*}
using the following interpolation: 
\begin{align*}
     \sum^n_{k=1} d_k \frac{\sum^{q}_{o=-q}  \sinc\left(r-o\right) w(o)  \langle \psi_1 \vert \prod^{M}_{j} V_j S^{o+m_k}_{p,j}(s_k t_{j})  \vert \psi_2 \rangle }{w(r_k)},
\end{align*}
where $w(t) = \frac{1}{\sigma \sqrt{2\pi}}\exp\left(-\frac{t^2}{2{\sigma^2}}\right)$, $\sigma=\sqrt{\frac{q+2}{\pi}}$, $m_k=\lfloor g/s_k \rfloor$, $r_k=g/s_k - \lfloor g/s_k \rfloor$, $g=T_j/t_j$, $s_i = \cos\left(\frac{2i-1}{2n} \pi\right)$, and $d_k = \frac{1}{n}(-1)^{k+n/2} \tan\left(\frac{2k-1}{2n}\pi\right)$. Here, 
\begin{align*}
    \alpha_j = \sum_{(\upsilon,m)} \acomm\big(H_{j,\pi_{\Upsilon}(\Gamma)},\ldots,H_{j,\pi_{\upsilon}(m+1)},H_{j,\pi_{\upsilon}(m)}\big),
\end{align*}
$T_{\rm max} = \max_j T_j$ and $\alpha_{\rm max} = \max_j \alpha_j$. This is, while choosing $n  = \tilde{O}\left(\log{\frac{1}{\epsilon}}\right)$, $q= O\left(\log{\frac{1}{\epsilon}} \right)$, 
\begin{align*}
    g = \Theta\left(\left(  \frac{M \alpha_{\text{max}} T_{\text{max}}^{p+1} }{p!} \right)^{\frac{1}{p}}\right),
\end{align*}
and
\begin{align*}
    p =\Theta\left(\sqrt{\frac{\log{M T_{\rm max} }}{\log{5}}}\right),
\end{align*}
and $\sum_j t_j  \leq \pi/2$. Also, this is provided that $\vert \prod^{M}_{j} V_j S^{o+m_k}_{p,j}(s_k t_{j})  \vert \psi_2 \rangle$ are estimated with $O\left(\varepsilon/\log^2 (1/\epsilon)\right)$ precision through amplitude estimation.
\end{theorem}

Now, we look at the case $M=1$:

\begin{corollary}

There exists a quantum algorithm for Hamiltonian evolution whose query cost scales like $\tilde{O}\left(\min\left(T^{1+o(1)}\alpha^{o(1)}\log^2 (1/\epsilon)/\varepsilon,1\right)\right)$ and a maximum query depth that scales like $\tilde{O}\left(\max\left(T ^{1+o(1)}\alpha^{o(1)}\log(1/\epsilon)/\varepsilon,1\right)\right)$, while only using two extra ancilla qubits.
\end{corollary}

One extra qubit requirement is coming from the Hadamard test to block-encode both the real and imaginary parts of the complex amplitudes, and the other extra qubit is for the amplitude estimation algorithm.

\section{Discussion and outlook}

We have provided an alternative method for the calculation of fractional queries by using a cardinal sine interpolation in conjunction with amplitudes of integer power queries. The algorithmic error from the query interpolation alone converges exponentially with the number of amplitudes used. Cardinal sine interpolation also proves to be very stable and does not propagate uncertainty significantly when using confidence intervals or normally-distributed random variables. Provided also here is a method to interpolate Trotterized evolution to zero Trotter-step-size.

The scaling with all relevant parameters is almost optimal for static Hamiltonian evolution. The expected quadratic scaling $\log^2(1/\epsilon)$ of the method presented here is still asymptotically worse than the state-of-the-art $\log(1/\epsilon)$. Extending these kind of methods for time-dependent simulation is also of great importance, however, no attempts at this were made in this work. Quantum signal processing methods are not good at this either, however there are some new methods available using virtual clock registers and multi-product formulas~\cite{watkins2024time} which achieve ${\rm polylog}(1/\epsilon)$. More importantly, the algorithm obtained here has, for the first time, an exponentially small algorithmic error in terms of computational resources while having a quasi-linear cost-scaling with respect to the evolution time compared to similar methods using product formulas plus stable interpolation~(see for example \cite{watson2024exponentially}). This was allowed by the variable order in $p$ we proposed. This could be quite possibly extended to other works using stable interpolation with product formulas.

The results here, combined with low-weight Fourier expansion, Lemma 37, from \cite{van2020quantum} or Fourier series accelerators/mollifiers (\cite{Vandeven1991}) open up the possibility for efficient off-loading of a great extent of quantum signal processing into a classical post-processing (for the $M=O(1)$ case) for Hermitian operators. Such examples include Gibbs sampling ($f(H)=\exp(-\beta H)$), and matrix inversion $f(H) = H^{-1}$. This, while keeping competent cost scaling with respect to state-of-the-art methods. All that is left for the quantum computer is to estimate the matrix element amplitudes for unitary operators. If the amplitude of these operators with respect to two quantum states is expected to be small, one is better off estimating these coherently as there is a quadratic advantage lost if one estimates the amplitudes of unitaries incoherently and then added up. In this regard, quantum signal processing remains better as well.

\section{Acknowledgements}

The author would like to acknowledge the useful comments by Rutuja Kshirsagar, Bibhas Adhikari, and Taozhi Guo for their useful comments on the earlier versions. GR would also like to thank James Watson and Jacob Watkins for pointing out an important error introduced in the second version of the manuscript. In particular GR would like to thank Jacob Watkins for his generous input on the write-up of the overall manuscript. GR thanks the 2025 QIP reviewers for useful feedback as well as the useful feedback from Quantum Journal reviewers. The author used large language models to assist with language editing, formatting, and bibliography preparation; all scientific content, derivations, and results are the author's own, and the author takes full responsibility for the final manuscript.


\begin{thebibliography}{10}

\bibitem{aaronson2006TenMostAnnoying}
S.~Aaronson, ``The ten most annoying questions in quantum computing, 2006.''
  \url{https://www.scottaaronson.com/blog/?p=112}.

\bibitem{Sheridan_fractional}
L.~{Sheridan}, D.~{Maslov} and M.~{Mosca}, \emph{{Approximating fractional time
  quantum evolution}},
  \href{https://doi.org/10.1088/1751-8113/42/18/185302}{\emph{Journal of
  Physics A Mathematical General} {\bfseries 42} (2009) 185302}
  [\href{https://arxiv.org/abs/0810.3843}{{\ttfamily 0810.3843}}].

\bibitem{Cleve_2009}
R.~Cleve, D.~Gottesman, M.~Mosca, R.D.~Somma and D.~Yonge-Mallo,
  \emph{Efficient discrete-time simulations of continuous-time quantum query
  algorithms},  in \emph{Proceedings of the forty-first annual ACM symposium on
  Theory of computing}, STOC ’09, p.~409–416, ACM, May, 2009,
  \href{https://doi.org/10.1145/1536414.1536471}{DOI}.

\bibitem{Berry_Hamiltonian_sim_2013}
D.W.~Berry, A.M.~Childs, R.~Cleve, R.~Kothari and R.D.~Somma, \emph{Exponential
  improvement in precision for simulating sparse {H}amiltonians},  in
  \emph{Proceedings of the Forty-Sixth Annual ACM Symposium on Theory of
  Computing}, STOC '14, pp.~283--292, ACM, 2014,
  \href{https://doi.org/10.1145/2591796.2591854}{DOI}
  [\href{https://arxiv.org/abs/1312.1414}{{\ttfamily 1312.1414}}].

\bibitem{Low_2017}
G.H.~Low and I.L.~Chuang, \emph{Optimal hamiltonian simulation by quantum
  signal processing},
  \href{https://doi.org/10.1103/physrevlett.118.010501}{\emph{Physical Review
  Letters} {\bfseries 118} (2017) }.

\bibitem{low2018hamiltonian}
G.H.~Low and N.~Wiebe, \emph{Hamiltonian simulation in the interaction
  picture},  2018.
\newblock 10.48550/arXiv.1805.00675.

\bibitem{gilyen_SVT}
A.~Gily{\'e}n, Y.~Su, G.H.~Low and N.~Wiebe, \emph{Quantum singular value
  transformation and beyond: exponential improvements for quantum matrix
  arithmetics},  in \emph{Proceedings of the 51st Annual ACM SIGACT Symposium
  on Theory of Computing}, STOC 2019, pp.~193--204, ACM, 2019,
  \href{https://doi.org/10.1145/3313276.3316366}{DOI}
  [\href{https://arxiv.org/abs/1806.01838}{{\ttfamily 1806.01838}}].

\bibitem{Rendon_extrapol}
G.~Rendon, J.~Watkins and N.~Wiebe, \emph{Improved accuracy for {T}rotter
  simulations using {C}hebyshev interpolation},
  \href{https://doi.org/10.22331/q-2024-02-26-1266}{\emph{Quantum} {\bfseries
  8} (2024) 1266} [\href{https://arxiv.org/abs/2212.14144}{{\ttfamily
  2212.14144}}].

\bibitem{whittaker1915xviii}
E.T.~Whittaker, \emph{{XVIII}.---on the functions which are represented by the
  expansions of the interpolation-theory},
  \href{https://doi.org/10.1017/S0370164600017806}{\emph{Proceedings of the
  Royal Society of Edinburgh} {\bfseries 35} (1915) 181}.

\bibitem{kotel2006transmission}
V.A.~Kotel'nikov, \emph{On the transmission capacity of `ether' and wire in
  electric communications},
  \href{https://doi.org/10.1070/PU2006v049n07ABEH006160}{\emph{Physics-Uspekhi}
  {\bfseries 49} (2006) 736}.

\bibitem{nyquist1928certain}
H.~Nyquist, \emph{Certain topics in telegraph transmission theory},
  \href{https://doi.org/10.1109/T-AIEE.1928.5055024}{\emph{Transactions of the
  American Institute of Electrical Engineers} {\bfseries 47} (1928) 617}.

\bibitem{shannon1948mathematical}
C.E.~Shannon, \emph{A mathematical theory of communication},
  \href{https://doi.org/10.1002/j.1538-7305.1948.tb01338.x}{\emph{The Bell
  System Technical Journal} {\bfseries 27} (1948) 379}.

\bibitem{Suzuki1991GeneralTO}
M.~Suzuki, \emph{General theory of fractal path integrals with applications to
  many-body theories and statistical physics},
  \href{https://doi.org/10.1063/1.529425}{\emph{Journal of Mathematical
  Physics} {\bfseries 32} (1991) 400}.

\bibitem{Low2019hamiltonian}
G.H.~Low and I.L.~Chuang, \emph{Hamiltonian {S}imulation by {Q}ubitization},
  \href{https://doi.org/10.22331/q-2019-07-12-163}{\emph{{Quantum}} {\bfseries
  3} (2019) 163}.

\bibitem{Rendon2021}
G.~Rendon, T.~Izubuchi and Y.~Kikuchi, \emph{Effects of cosine tapering window
  on quantum phase estimation},
  \href{https://doi.org/10.1103/PhysRevD.106.034503}{\emph{Physical Review D}
  {\bfseries 106} (2022) 034503}
  [\href{https://arxiv.org/abs/2110.09590}{{\ttfamily 2110.09590}}].

\bibitem{watson2024exponentially}
J.D.~Watson and J.~Watkins, \emph{Exponentially reduced circuit depths using
  {T}rotter error mitigation},
  \href{https://doi.org/10.1103/kw39-yxq5}{\emph{PRX Quantum} {\bfseries 6}
  (2025) 030325} [\href{https://arxiv.org/abs/2408.14385}{{\ttfamily
  2408.14385}}].

\bibitem{low2019well}
G.H.~Low, V.~Kliuchnikov and N.~Wiebe, \emph{Well-conditioned multiproduct
  {H}amiltonian simulation},  2019.
\newblock 10.48550/arXiv.1907.11679.

\bibitem{watkins2024time}
J.~Watkins, N.~Wiebe, A.~Roggero and D.~Lee, \emph{Time-dependent {H}amiltonian
  simulation using discrete-clock constructions},
  \href{https://doi.org/10.1103/PRXQuantum.5.040316}{\emph{PRX Quantum}
  {\bfseries 5} (2024) 040316}.

\bibitem{2000quant.ph..5055B}
G.~Brassard, P.~H{\o}yer, M.~Mosca and A.~Tapp, \emph{Quantum amplitude
  amplification and estimation},  in \emph{Quantum Computation and
  Information}, vol.~305 of \emph{Contemporary Mathematics}, pp.~53--74,
  American Mathematical Society (2002),
  \href{https://doi.org/10.1090/conm/305/05215}{DOI}
  [\href{https://arxiv.org/abs/quant-ph/0005055}{{\ttfamily
  quant-ph/0005055}}].

\bibitem{IQAE}
D.~{Grinko}, J.~{Gacon}, C.~{Zoufal} and S.~{Woerner}, \emph{{Iterative quantum
  amplitude estimation}},
  \href{https://doi.org/10.1038/s41534-021-00379-1}{\emph{npj Quantum
  Information} {\bfseries 7} (2021) 52}
  [\href{https://arxiv.org/abs/1912.05559}{{\ttfamily 1912.05559}}].

\bibitem{Gautschi90how(un)stable}
W.~Gautschi, \emph{How (un)stable are {V}andermonde systems?},  in
  \emph{Asymptotic and Computational Analysis}, R.~Wong, ed., vol.~124 of
  \emph{Lecture Notes in Pure and Applied Mathematics}, (New York),
  pp.~193--210, Marcel Dekker, 1990,
  \href{https://doi.org/10.1201/9781003072584-9}{DOI}.

\bibitem{Kuian_cheb}
M.~Kuian, L.~Reichel and S.~Shiyanovskii, \emph{Optimally conditioned
  vandermonde-like matrices},
  \href{https://doi.org/10.1137/19M1237272}{\emph{SIAM Journal on Matrix
  Analysis and Applications} {\bfseries 40} (2019) 1399}
  [\href{https://arxiv.org/abs/https://doi.org/10.1137/19M1237272}{{\ttfamily
  https://doi.org/10.1137/19M1237272}}].

\bibitem{mason2002chebyshev}
J.C.~Mason and D.C.~Handscomb, \emph{Chebyshev Polynomials}, Chapman and
  Hall/CRC (2002),
  \href{https://doi.org/10.1201/9781420036114}{10.1201/9781420036114}.

\bibitem{trefethen2019approximation}
L.N.~Trefethen, \emph{Approximation Theory and Approximation Practice, Extended
  Edition}, SIAM, Philadelphia, PA (2019),
  \href{https://doi.org/10.1137/1.9781611975949}{10.1137/1.9781611975949}.

\bibitem{van2020quantum}
J.~van Apeldoorn, A.~Gily{\'e}n, S.~Gribling and R.~de~Wolf, \emph{Quantum
  {SDP}-solvers: Better upper and lower bounds},
  \href{https://doi.org/10.22331/q-2020-02-14-230}{\emph{Quantum} {\bfseries 4}
  (2020) 230}.

\bibitem{Vandeven1991}
H.~Vandeven, \emph{Family of spectral filters for discontinuous problems},
  \href{https://doi.org/10.1007/BF01062118}{\emph{J. Sci. Comput.} {\bfseries
  6} (1991) 159–192}.

\end{thebibliography}

\providecommand{\href}[2]{#2}\begingroup\raggedright\endgroup

\appendix

\section{Uncertainty propagation\label{app:sinc_stability}}

In this section, we provide bounds on the uncertainty propagation to the cardinal sine interpolant coming from errors on integer queries estimates. First, we estimate how the uncertainty propagates when the estimates come in the form of confidence intervals. This is summarized in the following lemma:

\begin{lemma}\label{lem:uncert_prop}
    Given an error vector $\epsilon\in \mathbb{C}^{2q+1}$, where $\|\epsilon\|_{\infty}\leq \varepsilon_{\rm estimate}$. The error tolerance on the interpolated quantity is 
    \begin{align}
    \max_t\| \epsilon s(t) \|_1 \leq  \varepsilon_{\rm estimate} \left( 3 + \frac{2}{\pi} \log\left(2q\right) \right).
\end{align} 
Here, 
\begin{align}
     \|s(t)\|_1 &= \sum^{q}_{n=-q}\left|\sinc\left(t-n\right)\right|
\end{align}
\end{lemma}

\begin{proof}
Using H\"older's inequality, one obtains that the error tolerance on the interpolated quantity is
\begin{align}
    \max_t\| \epsilon s(t) \|_1 \leq \|\epsilon\|_{\infty} \max_t\|s(t)\|_{1}.
\end{align}
Now, we would like to upper-bound $\max_t\|s(t)\|_1$. First note
\begin{align}
    \max_t\|s(t)\|_1 &= \sum^{q}_{n=-q}\left|\sinc\left(\tau-n\right)\right|.
\end{align}
With this, we note that one can split the sum into two parts
\begin{align}
    \max_\tau\|s(\tau)\|_1&\leq  \sum^{q}_{n=\lfloor \tau \rfloor+1}\left|\sinc\left(n-\tau\right)\right| + \sum^{\lfloor\tau\rfloor}_{n=-q}\left|\sinc\left(n-\tau\right)\right|
\end{align}
We now upper-bound the terms $|\sinc(\lfloor \tau \rfloor + 1 - \tau)|$, $|\sinc(\lfloor \tau \rfloor - \tau)|$, and $|\sinc(\lfloor \tau \rfloor - 1 - \tau)|$ with $1$, obtaining
\begin{align}
    \max_\tau\|s(\tau)\|_1&\leq 3 + \sum^{q}_{n=\lfloor \tau \rfloor+2}\left|\sinc\left(n-\tau\right)\right| + \sum^{\lfloor\tau\rfloor-2}_{n=-q}\left|\sinc\left(n-\tau\right)\right|\cr
\end{align}
We can now upper bound the summands with
\begin{align}
   \max_\tau\|s(\tau )\|_1 &\leq 3 + \frac{1}{\pi}\sum^{q}_{n=\lfloor \tau \rfloor+2}\abs{\frac{1}{n-\tau}} + \frac{1}{\pi}\sum^{\lfloor\tau\rfloor-2}_{n=-q}\abs{\frac{1}{n-\tau}} \cr 
    &\leq 3 + \frac{1}{\pi}\sum^{q}_{n=-q+2}\abs{\frac{1}{n+q}} + \frac{1}{\pi}\sum^{q-2}_{n=-q}\abs{\frac{1}{n-q}} \cr 
    &= 3 + \frac{2}{\pi}\sum^{2q}_{n=2}\frac{1}{n} .    
\end{align}
Following that, we use right-Riemann sum approximation of an integral to upper bound the last term, finally getting
\begin{align}
    \max_\tau\|s(\tau)\|_1&\leq 3 +  \frac{2}{\pi}\int^{2q}_{1}\frac{1}{x} \mathrm{d}x \cr
    &\leq 3 + \frac{2}{\pi} \log\left(2q\right).
\end{align}
\end{proof}

Now, if we were to estimate the amplitudes for integer queries sampling a Gaussian distribution, the variance on the interpolated quantity would be

\begin{lemma}\label{lem:uncert_prop_gauss}
    Suppose one has $2q+1$ independent random variables $X_n \sim \mathcal{N}(\mu_n,\sigma_n^2)$, where $n\in\{-q,-q+1,\dots,q\}$. One has that the interpolating quantity is the random variable
    \begin{align*}
        X_{\rm interp}(t) =  \sum^q_{n=-q} \sinc \left(t-n\right)X_n.
    \end{align*}
    For such, the variance $\sigma^2(t)$ is calculated through
    \begin{align*}
        \sigma(t)^2 = s^T(t) C s(t)
 \end{align*}
    where $C = {\rm diag} (\sigma^2_{-q},\sigma_{-q+1}^2,\dots,\sigma_q^2)$ and $\left(s(t)\right)_n = \sinc\left(t-n\right)$ and it can be bounded through
 \begin{align*}
     \sigma(t)^2 \leq \left(3+\frac{4}{\pi^2}\right) \max_n (\sigma_n^2).
 \end{align*}
\end{lemma}
\begin{proof}
\begin{align}
    \sigma(t)^2 &= s^{T}(t) C_x s(t) \cr
    &= \| \sigma^2 s^2 \|_1 \leq \| \sigma^2 \|_\infty \| s^2 \|_1 = \| \sigma^2 \|_\infty \| s \|^2
\end{align}

\begin{align}
    \|s\|^2&= \sum^q_{n=-q}\left|\sinc\left(t-n\right)\right|^2 \cr 
    &\leq\sum^{\infty}_{n=-\infty}\left|\sinc\left(t-n\right)\right|^2 \cr 
    &\leq\max_{\tau\in[-1,1]}\sum^{\infty}_{n=-\infty}\left|\sinc\left(\tau-n\right)\right|^2 \cr 
    &\leq \max_{\tau\in[-1,1]}\left(\sinc^2\left(\tau\right) + 2\sum^{\infty}_{n=1}\left|\sinc\left(n-\tau\right)\right|^2\right)\cr 
    &\leq \max_{\tau\in[-1,1]}\left(\left(\sinc^2\left(\tau\right)+2\sinc^2(\tau-1)\right)+\frac{2}{\pi^2}\sum^{\infty}_{n=2}\frac{1}{(n-\tau)^2} \right) \cr
    &\leq 3+\frac{4}{\pi^2}   
\end{align}
\end{proof}

\end{document}